\documentclass{jalc}

\usepackage{graphicx}
\usepackage{complexity}
\usepackage{amssymb}
\usepackage{wasysym}

\usepackage{endnotes}
\usepackage{listings}
\usepackage{mathtools}
\usepackage{verbatim}
\usepackage{mathrsfs}
\usepackage{shuffle}
\usepackage{longtable}
\usepackage{enumitem}
\usepackage{etoolbox}
\usepackage{float}
\usepackage{color}

\newclass{\NTM}{NTM}
\newclass{\NCM}{NCM}
\newclass{\DCM}{DCM}
\newclass{\DTM}{DTM}
\newclass{\NFA}{NFA}
\newclass{\NPDA}{NPDA}
\newclass{\DPDA}{DPDA}
\newclass{\DFA}{DFA}

\newsavebox{\spacebox}
\begin{lrbox}{\spacebox}
\verb*! !
\end{lrbox}
\newcommand{\blank}{\usebox{\spacebox}}%
\newcommand\rw{\downarrow}

\newenvironment{claimproof}{%
  \ifnum \value{d@proof}=0{\setcounter{claim}{0}}\else\fi
  \stepcounter{d@proof}\par\noindent
  {\rmfamily\itshape\mdseries Proof.}\hspace{\labelsep}\ignorespaces}%
  {\addtocounter{d@proof}{-1}\mbox{}\nolinebreak\hfill~%
  \ifnum \value{d@proof}=0{\qed}\else
    \ifnum \value{d@proof}=1{\qed\nolinebreak\,\nolinebreak}\else
      \ifnum \value{d@proof}=2{\qed\nolinebreak\,\nolinebreak\qed
          \nolinebreak\,\nolinebreak\qed}\else
        {\qed\nolinebreak...\nolinebreak\qed}%
  \fi\fi\fi
  \medbreak
}

\newtheorem{construction}[theorem]{Construction}

\begin{document}%
     %
\title{On the Density of Languages Accepted by Turing Machines and Other Machine Models}
     %

\runningtitle{Density of Turing Machines}
     
\runningauthors{\textsc{O.H.~Ibarra}, \textsc{I.~McQuillan}}
     

\author[SB]{Oscar H. Ibarra}
\address[SB]{Department of Computer Science,
  University of California, Santa Barbara, \\
  CA 93106, USA\\
  \email{ibarra@cs.ucsb.edu}
}

\author[UofS,IAN]{Ian McQuillan}
\address[UofS]{Department of Computer Science,
  University of Saskatchewan,\\
  110 Science Place,
  Saskatoon, SK,
  Canada\\
  \email{mcquillan@cs.usask.ca}
}
\thanks[IAN]{Supported in part by a grant from NSERC.}
\thanks{Published at Journal of Languages, Automata and Combinatorics, 23, 2018, 189--199 DOI: 10.25596/jalc-2018-189}

\maketitle
%
\begin{abstract}
A language is dense if the set of all infixes (or subwords) of the language is the set of all words. Here, it is shown that it is decidable whether the language accepted by a nondeterministic Turing machine with a one-way read-only input and a reversal-bounded read/write worktape (the read/write head changes direction at most some fixed number of times) is dense. From this, it is implied that it is also decidable for one-way reversal-bounded queue automata, one-way reversal-bounded stack automata, and one-way reversal-bounded $k$-flip pushdown automata (machines that can ``flip'' their pushdowns up to $k$ times). However, it is 
undecidable for deterministic Turing machines with
two 1-reversal-bounded worktapes (even when the two tapes are restricted to operate as 1-reversal-bounded pushdown stacks).
\keywords
    density, Turing machines, store languages, pushdowns, queues, stacks
\end{abstract}

%
%
%
\section{Introduction}

A language $L \subseteq \Sigma^*$ is said to be {\em dense} if
the set of all infixes of $L$ is equal to $\Sigma^*$ \cite{TheoryOfCodes}.
This is an interesting property especially relevant to the theory of codes \cite{CodesHandbook}.
The notion has been investigated as it pertains to independent sets,
maximal independent sets, and disjunctive languages \cite{denseIto,shyr}.
Later, the notion was generalized from the set of infixes of a language being the universe, to arbitrary relations used in place of the infix relation \cite{JKT}. For example, a language $L$ is suffix-dense if
the set of all suffixes of $L$ is equal to the universe.
Homomorphisms that preserve different types of density were investigated
as well \cite{JurgensenActaCybernetica}.

Recently, these generalized notions of density were studied
as applied to types of pushdown automata and counter machines \cite{densityCFL}. 
It was surprisingly found that it is decidable whether a language
accepted by a reversal-bounded nondeterministic pushdown automaton ($\NPDA$) is dense (a pushdown is $l$ reversal-bounded if the number of changes in direction between pushing and popping is at most $l$, and is reversal-bounded if it is $l$ reversal-bounded for some $l$). This is true despite the fact that it is undecidable whether the language accepted by such a machine is equal to $\Sigma^*$ (even if the pushdown is a 1-reversal-bounded counter). Therefore, deciding whether the set of infixes is the universe is {\it easier} than deciding whether the set itself is equal to the universe. 
However, density was found to be undecidable for deterministic
pushdown automata (without a reversal-bound), and nondeterministic
one counter automata (again without a reversal-bound).
Decidability and undecidability results for other variants of density are also presented in \cite{densityCFL}.

A key property used to prove decidability of density for reversal-bounded $\NPDA$s was that of the store language. The store
language of a pushdown is the set of all state/store content
pairs, $qx$ where $q$ is a state and $x$ is a word over the pushdown
alphabet, that can appear in any accepting computation. It
is known that the store language of every $\NPDA$ is a regular language
\cite{CFHandbook}. This inspired the authors to investigate
the store language of other machine models, especially some
types of Turing machines \cite{StoreTechReport}. Of particular
interest was the recent result that the store language of every
nondeterministic Turing machine with a one-way read-only input tape,
and a reversal-bounded worktape (the store; where there is a bound
on the number of changes of direction of the read/write head) is
a regular language.

In this paper, density of languages accepted by this same type of nondeterministic Turing machines is similarly shown to be decidable using the regularity of the store languages.  A corollary to this is
decidability of density for languages accepted by several types of one-way nondeterministic reversal-bounded machine models: stack automata \cite{StackAutomata,Stack2} (machines that can enter the ``contents'' of the pushdown in read-only mode),
queue automata \cite{Harju2002278}, and $k$-flip pushdown automata \cite{flipPushdown} (machines that can
flip their pushdown contents at most $k$ times).
However, undecidability is obtained for one-way deterministic Turing machines with a one-way 
read-only input and two 1-reversal-bounded pushdowns,
or for nondeterministic Turing machines over a unary alphabet.
Decidability of density for several types of two-way machine
models is also investigated.

\section{Preliminaries}

In this paper, we assume knowledge of formal language and automata
theory, referring to \cite{HU} for an introduction. This includes
nondeterministic finite automata ($\NFA$), nondeterministic
pushdown automata (\NPDA), nondeterministic Turing machines, and
their deterministic variants (obtained by replacing $N$ with $D$).

An alphabet $\Sigma$ is a finite set of symbols, and $\Sigma^*$ is
the set of all words over $\Sigma$. A language $L$ is any subset of
$\Sigma^*$. Let $w \in \Sigma^*$. Then $|w|$ is the length
of $w$, and $|w|_a$ is the number of $a$'s in $w$, for $a \in \Sigma$. 
The reverse of $w$, $w^R$ is the word obtained by
reversing the order of the letters of $w$.
A word $y$ is an infix of $w$ if $w = xyz$,
for some $x,z \in\Sigma^*$. Given $L \subseteq \Sigma^*$, 
$\inf(L) = \{y \mid y \mbox{~is an infix of~} w \in L\}$.
The left quotient of $L$ by $R$, $R^{-1} L = \{y \mid xy \in L, x \in R\}$.

Next, we will define nondeterministic Turing machines, which
we will define to have a one-way read-only input, and a separate
read/write bi-infinite worktape. 

A nondeterministic Turing machine ($\NTM$) is a tuple
$M = (Q, \Sigma, \Gamma, \blank, \downarrow, \delta, q_0, F)$, where
$Q$ is the finite set of states, $\Sigma$ is the input alphabet,
$\Gamma$ is the worktape alphabet, $\blank \in \Gamma$ is the blank symbol, $\downarrow \in \Gamma$ is the worktape
head, $q_0 \in Q$ is the initial
state, $F \subseteq Q$ is the set of final states,
$\Gamma_0 = \Gamma - \{\downarrow\}$,
and
$\delta$ is a relation from 
$Q \times (\Sigma \cup \{\lambda\}) \times \Gamma_0$ to
$Q \times \Gamma_0 \times \{{\rm L}, {\rm S}, {\rm R}\}$.

A {\em configuration} of $M$ is a tuple
$(q, w, x)$, where $q \in Q$ is the current state,
$w \in \Sigma^*$ is the remaining input, and 
$x \in \Gamma^*$ is the worktape contents with $|x|_{\downarrow}= 1$.
Next, we will describe how configurations change. Below,
$q, q' \in Q, a \in \Sigma \cup \{\lambda\}, x \in (\Gamma_0^*- \blank \Gamma_0^*),
y \in (\Gamma_0^*- \Gamma_0^* \blank), c,d,c', d' \in \Gamma_0$.
\begin{itemize}
\item $(q,aw,x \downarrow c y) \vdash_M (q', w, x \downarrow d y)$, if
$(q',d, S) \in \delta(q, a, c),$
\item $(q,aw,x \downarrow c y) \vdash_M (q', w, x d' \downarrow y')$, if
$(q',d, R) \in \delta(q, a, c), (y = \lambda \implies y' = \blank, \mbox{~otherwise~} y' = y), (x = \lambda \wedge d = \blank \implies d' =\lambda, \mbox{~otherwise~} d' = d)$,
\item $(q,aw,x \downarrow c y) \vdash_M (q', w, x' \downarrow c' d' y)$, if
$(q',d, L) \in \delta(q, a, c), (x = \lambda \implies x' = \lambda, c' = \blank, \mbox{~otherwise~} x = x' c'), (y = \lambda \wedge d = \blank \implies d' =\lambda, \mbox{~otherwise~} d' = d)$.
\end{itemize}
Let $\vdash_M^*$ be the reflexive and transitive closure of $\vdash_M$.
Then, the language accepted by $M$, denoted by $L(M)$ is defined to be:
$$L(M) = \{ w \mid (q_0, w, \downarrow \blank) \vdash_M^*
(q_f, \lambda, x), w \in \Sigma^*, q_f \in F\}.$$
Furthermore, the store language of $M$, $S(M)$, is defined to be:
$$S(M) = \{ qx \mid (q_0, w, \downarrow \blank) \vdash_M^*
(q, w', x) \vdash_M^* (q_f, \lambda, x'), q_f \in F,
w, w' \in \Sigma^*\}.$$

An $\NTM$ of this type is said to be {\em reversal-bounded}, if there
is a $k$ such that $M$ makes at most $k$ changes between moving the
worktape left and right on every accepting computation.
Such a machine can be assumed to be deterministic in the usual
fashion \cite{StoreTechReport}.

In \cite{StoreTechReport}, the following was shown:
\begin{proposition}
\label{TMReg}
Given a nondeterministic Turing machine $M$ with a one-way read-only input tape, and a reversal-bounded worktape, then $S(M)$ is a regular language that can be
effectively constructed from $M$.
\end{proposition}
It is clear that should the reversal-bounded condition be removed, then
$S(M)$ could be an arbitrary recursively enumerable language
after a left quotient by a state symbol, since, given an arbitrary Turing machine with a two-way read/write
input tape $M$ and initial state $q_0$ (that without loss of generality
is never re-entered), a TM with a one-way read-only input tape and a worktape $M'$ could be constructed that copies the input to the store, then simulates $M$ on the worktape. Then $(q_0\downarrow)^{-1}S(M) = L(M)$.

\section{Density of Turing Machines}

A language $L \subseteq \Sigma^*$ is {\em dense} if
$ \inf(L) = \Sigma^*$. Even though it has long
been known that it is undecidable for even
a one-way nondeterministic one counter automaton that makes only
one reversal on the counter, whether the language it accepts is equal to $\Sigma^*$ \cite{Baker1974}, 
the infix
operator perhaps counter-intutively makes the problem easier.

Here, we will prove that it is decidable whether a language accepted by a
one-way nondeterministic Turing machine with a reversal-bounded read/write 
worktape is dense. Then, 
it will follow that this is true  for a number of different machine models.
The main tool of the proof is that for these types of Turing machines,
the store language is a regular language, by Proposition \ref{TMReg}.
\begin{proposition}
It is decidable, given $L$ accepted by a $\NTM$ with a one-way read-only
input and a $k$-reversal-bounded read/write worktape, whether
$L$ is dense.
\end{proposition}
\begin{proof}
Let $M = (Q,\Sigma,\Gamma, \blank,\downarrow, \delta,q_0,F)$ be an 
$\NTM$ with a one-way input and a 
$k$ reversal-bounded worktape. 
Assume without loss of generality that all transitions that stay on
the worktape do not change the worktape (any changes can be remembered in the state and applied when moving).

Next, 
assume without loss of generality that the states of $M$ 
are partitioned into subsets 
$Q_0, Q_1, \ldots, Q_k$,
where $Q_i$ contains all states defined on or after the $i$th reversal, but before the $(i+1)$st reversal. 
Let $T$ be a set of labels in bijective correspondence with the
transitions of $M$ (each transition $(p',d,x) \in \delta(p,a,c)$ has an 
associated label in $T$).

Consider $S(M)$, which is a regular language by Proposition \ref{TMReg}.
For each $q \in Q$, consider $R^q = q^{-1} S(M)$, also regular.

Consider the language 
$$L^q = \{w \mid (q_0, uwv , \rw \blank) \vdash_M^* (q, wv , \alpha) \vdash_M^* (q, v, \beta) \vdash_M^* (q_f, \lambda, \gamma), u,v,w \in \Sigma^*, q_f \in F\}.$$

It will be proven that $L^q$ is regular for each $q \in Q$. 

A sequence of transition labels $y = t_1 \cdots t_m, t_i \in T, 1 \leq i \leq m$, is {\em valid for state $q$} if $t_1$ starts in state $q$,
and $t_m$ switches to $q$, and for each $i$, $1 \leq i < m$, the
outgoing state of $t_i$ is the same as the incoming state of $t_{i+1}$,
and $t_i$ being a transition that stays on letter $d \in \Gamma$ of the store implies $t_{i+1}$ also reads $d$. There is no restriction on the input word read while
transitioning via $y$.
Given a valid $y$, let $\overleftarrow{y} \in \Gamma^*$ be the
word with the first 
letter of it being the store letter
$t_1$ is defined on, and subsequent letters are 
obtained from $y$ from left-to-right by concatenating the letters
{\it read} from the store after a transition that moves either left or right (but not stay) on the store.
Also, given a valid $y$, let $\overrightarrow{y}$ be obtained from
$y$ from left-to-right by concatenating all letters {\it written}
on the store during a transition
that moves either left or right (but not stay) on the store. 
Notice that since $\overleftarrow{y}$ and $\overrightarrow{y}$ are only defined
on valid $y$, since $y$ transitions from state $q$ to itself, and since the states
are partitioned by reversal-bounds, then $\overleftarrow{y}$ and $\overrightarrow{y}$ are
only defined on $y$ such that all transitions in $y$ only move right or stay, or move left or stay.
Lastly, let $\dot{y}$ be $\lambda$ if $y$
ends with a transition that moves right or left on the store,
and the last store letter read by the last transition of $y$ otherwise (if a stay transition).

Let $h_{\Sigma}$ be a homomorphism from $(T \cup \Gamma)^*$ to $\Sigma^*$ that erases all letters of
$\Gamma$ and maps each $t \in T$ to the input in $\Sigma \cup \{\lambda\}$
read by $t$.

For $q \in Q_i$, with $i$ even, then there are no transitions
moving left on the store between states $q$ and $q$.
Then, an $\NFA$ $M'$ is created accepting an intermediate language
$L_1^q \subseteq \Gamma^* T^* \Gamma^*$, where $M'$ does
the following in parallel:
\begin{itemize}
\item verifies that the input is of the form
$\mu y \nu$, where $\mu,\nu \in \Gamma_0^*$,
$y = t_1 \cdots t_m, t_i \in T, 1 \leq i \leq m$, and that
$y$ is valid for $q$,
\item 
verifies that $\mu \rw \overleftarrow{y} \nu \in R^q$,
\item verifies that $\mu \overrightarrow{y} \rw \dot{y} \nu
\in R^q$.
\end{itemize}

Intuitively, $\mu \rw \overleftarrow{y} \nu \in R^q$ enforces that this word is in the store language, which implies
that there is some input word which can reach this configuration in an accepting computation. Then, $M'$ enforces
that $y$ encodes a valid change of configurations between the two configurations involving $q$, which is possible since
the worktape head only moves to the right.
Since $y$ is valid, the sequence of transitions
switches $\overleftarrow{y}$ to $\overrightarrow{y}$ in $M$.
So $M$ can reach $\mu \overrightarrow{y} \rw \dot{y} \nu$, which
is then verified to be in $R^q$,
enforcing that this is in the store language, 
so an input word can lead it to acceptance.

\begin{claim}
$h_{\Sigma}(L_1^q) = L^q$.
\end{claim}
\begin{claimproof}
``$\subseteq$'' Let $s\in h_{\Sigma}(L_1^q)$. Thus, there exists
$r \in \Gamma^* T^* \Gamma^*$ such that $h_{\Sigma}(r) = s$ and
$r \in L_1^q$. 
Then $r = \mu y \nu, \mu,\nu \in \Gamma^*, 
y \in T^*$, where $M'$
verifies $\mu  \rw \overleftarrow{y} \nu  \in R^q$.
Then $(q_0, u, \rw \blank) \vdash_M^* (q, \lambda, \mu \rw \overleftarrow{y} \nu)$ 
for some $u \in \Sigma^*$. Then,
$(q_0, u s, \rw \blank) \vdash_M^* (q, s, \mu \rw \overleftarrow{y} \nu)$.
Then, on each letter of $y$, since $y$ is valid, 
$M'$ simulates $M$ on the last letter read from 
$\Gamma$, reading each letter from $s \in \Sigma^*$ and replacing each letter of $\overleftarrow{y}$ with the next one from 
$\overrightarrow{y}$, while starting and finishing in state $q$. 
Assume first that the last letter of $y$ is a transition that moves
right on the store.
Thus, 
$$(q, s, \mu \rw \overleftarrow{y} \nu) \vdash_M^* (q, \lambda, \mu \overrightarrow{y} \rw \nu = \mu \overrightarrow{y} \rw \dot{y} \nu),$$
since $\dot{y} = \lambda$. Similarly if the last
letter of $y$ is a transition that stays on the store, then
$(q, s, \mu \rw \overleftarrow{y} \nu) \vdash_M^* (q, \lambda, \mu \overrightarrow{y} \rw \dot{y} \nu)$, where $\dot{y} \in \Gamma$.
In either case, then since $M'$ verified 
$\mu \overrightarrow{y} \rw \dot{y} \nu \in R^q$,
and this implies $(q, v, \mu \overrightarrow{y} \rw \dot{y} \nu) \vdash_M^* (q_f, \lambda, \gamma),
q_f \in F$, for some $v \in \Sigma^*$. Hence $h_{\Sigma}(y)=s \in L^q$.

``$\supseteq$''
Let $s \in L^q$. Thus,
\begin{equation}(q_0, usv , \rw \blank) \vdash_M^* (q, sv , \alpha_1 \rw \alpha_2 \alpha_4) \vdash_M^* (q,v ,\alpha_1 \alpha_3 \rw \alpha_4)
\vdash_M^* (q_f, \lambda, \gamma),
\label{claim1eqn}
\end{equation}
$q_f \in F$, for some $\alpha_1, \alpha_2, \alpha_3, \alpha_4 \in \Gamma_0^*$. 
Let the derivation above between states $q$ and $q$ be via transitions $t_1, \ldots, t_n$
respectively. Let $y = t_1 \cdots t_n$. First, $y$ must be valid for $q$. Assume first that this sequence ends with a transition that moves right on the store. 
Then $\alpha_2 = \overleftarrow{y}$ and $\alpha_3 = \overrightarrow{y}$
and $\dot{y} = \lambda$. Indeed,
$\alpha_1 \rw \alpha_2 \alpha_4 \in R^q$ by Equation (\ref{claim1eqn}),
and $\alpha_1 \rw \alpha_2 \alpha_4 = \alpha_1 \rw \overleftarrow{y} \alpha_4 \in R^q$. Further, $\alpha_1 \alpha_3 \rw \alpha_4 \in R^q$
by (\ref{claim1eqn}), and $\alpha_1 \alpha_3 \rw \alpha_4 = 
\alpha_1 \overrightarrow{y} \rw \dot{y} \alpha_4 \in R^q$.
Thus, $s \in h_{\Sigma}(L_1^q)$. Similarly if the sequence ends with a stay transition.
 \end{claimproof}

For $q \in Q_i$ with $i$ odd, then another language $L_2^q$ is created.
Hence, that there are no right transitions on the store defined between state
$q$ and itself. This case is similar in principal, however,
it is slightly trickier since the transitions work in a right-to-left
fashion, and because the read/write head is placed to the left of the
scanned symbol, thereby causing a mild complication at the beginning
and end of the word $y$.

Create an intermediate $2$-way $\NFA$ $M'$ accepting a language 
$L_2^q \subseteq \Gamma^* T^* \Gamma^*$, where $M'$ does all of
the following:
\begin{itemize}
\item verifies that the input is of the form
$\mu y \nu $, where $\mu,\nu \in \Gamma_0^*$,
$z = y^R = t_1 \cdots t_m, t_i \in T, 1 \leq i \leq m$, and that
$z$ is valid for $q$,
\item let $z'$ be obtained from $(\overleftarrow{z})^R$ by inserting
$\rw$ in the second last position; it verifies 
that $\mu z' \nu \in R^q$ (by using the reverse of $R^q$) as it reads $y$ (from right-to-left),
\item if $y$ starts with a symbol that stays on the store, then $M'$ verifies that $\mu \rw (\overrightarrow{z}\dot{z})^R \nu \in R^q$; 
otherwise, $M'$ verifies
that $\mu_1 \rw b (\overrightarrow{z})^R \nu \in R^q, b \in \Gamma, \mu = \mu_1 b$ if $ \mu \neq \lambda$, and $\mu_1 = \lambda, b = \blank$ if $ \mu = \lambda$.

\end{itemize}
\begin{claim}
$h_{\Sigma}(L_2^q) = L^q$.
\end{claim}
\begin{claimproof}
``$\subseteq$'' Let $s\in h_{\Sigma}(L_2^q)$. Thus, there exists
$r \in \Gamma^* T^* \Gamma^*$ such that $h_{\Sigma}(r) = s$ and
$r \in L_2^q$. 
Then $r = \mu y \nu, \mu,\nu \in \Gamma^*, 
y \in T^*, z = y^R$, where $M'$
verifies that $\mu z' \nu  \in R^q$ with $z'$ obtained from 
$(\overleftarrow{z})^R$ by inserting $\rw$ in the second last position.
Then $(q_0, u, \rw \blank) \vdash_M^* (q, \lambda, \mu z' \nu)$ 
for some $u \in \Sigma^*$. Then, 
$(q_0, u s, \rw \blank) \vdash_M^* (q, s, \mu z' \nu)$.
Then, on each letter of $y$ from right-to-left, since $z$ is valid, 
$M'$ simulates $M$ on the store letter read from 
$\Gamma$, reading each letter from $s$ and replacing each letter of $\overleftarrow{z}$ with the next one from 
$\overrightarrow{z}$, while starting and finishing in state $q$. 
If the first letter of $y$ is a left transition, then
$$(q, s, \mu z' \nu) \vdash_M^* (q, \lambda, \mu_1 \rw b (\overrightarrow{z})^R \nu),$$
where $\mu = \mu_1 b$ if $\mu \neq \lambda$, and $\mu_1 = \lambda,b= \blank$ if $\mu = \lambda$.
Then since $M'$ verified $\mu_1 \rw b (\overrightarrow{z})^R \nu \in R^q$,
this implies $(q, v, \mu_1 \rw b (\overrightarrow{z})^R \nu) \vdash_M^* (q_f, \lambda, \gamma),
q_f \in F$, for some $v \in \Sigma^*$. Hence $h_{\Sigma}(y) =s \in L^q$.
Similarly if the first letter of $y$ is a stay transition.

``$\supseteq$''
Let $s \in L^q$. Thus,
\begin{equation}
(q_0, usv , \rw \blank) \vdash_M^* (q, sv , \alpha_1 \alpha_2 \rw a \alpha_3) \vdash_M^* (q,v ,\alpha_1 \rw \alpha_4 \alpha_3)
\vdash_M^* (q_f, \lambda, \gamma),
\label{claim2eqn}
\end{equation}
$q_f \in F$, for some $\alpha_1, \alpha_2, \alpha_3, \alpha_4 \in \Gamma_0^*, a \in \Gamma_0$. 
Let the derivation above between states $q$ and $q$
be via transitions $t_1, \ldots, t_n$
respectively. Let $y = t_n \cdots t_1$, and $z = y^R$. First, 
$z$ must be valid for $q$.
Assume first that this ends with a left transition. 

Then, in Equation (\ref{claim2eqn}), let $\alpha_2'$ be such that
$\alpha_2 = b \alpha_2', b\in \Gamma_0$ if $\alpha_2 \neq \lambda$, and
$\alpha_2' = \lambda,b = \blank$ otherwise, and let
$\alpha_4'$ be such that
$\alpha_4 = b \alpha_4'$ (it must start with $b$ since $z$ ends with
a left transition).
Then $(\overleftarrow{z})^R = \alpha_2' a$ and 
$(\overrightarrow{z})^R = \alpha_4'$.
Indeed,
$\alpha_1 b \alpha_2' \rw a \alpha_3 \in R^q$ by Equation 
(\ref{claim2eqn}), and 
$\alpha_1 b \alpha_2' \rw a \alpha_3 = \alpha_1 b z' \alpha_3$,
where $z'$ is obtained from $(\overleftarrow{z})^R$ by inserting
$\rw$ in the second last position. Further,
$\alpha_1 \rw b \alpha_4' \alpha_3 \in R^q$ by
Equation (\ref{claim2eqn}), and $\alpha_1 \rw b \alpha_4' \alpha_3 =
\alpha_1 \rw b (\overrightarrow{z})^R \alpha_3 \in R^q$.
Thus, $s \in h_{\Sigma}(L_2^q)$. Similarly if the sequence ends with a stay transition.
 \end{claimproof}

Hence, $L^q$ is regular since regular languages are closed under homomorphism. Thus, $L' = \bigcup_{q\in Q} L^q$ is
also regular.

Lastly, it is shown that $\inf(L) = \Sigma^*$ if and only if $\inf(L') = \Sigma^*$. Indeed, it is immediate that if $\inf(L') = \Sigma^*$, then $\inf(L) = \Sigma^*$, since
$\inf(L') \subseteq \inf(L)$. For the opposite direction, assume $\inf(L) = \Sigma^*$.
Let $w \in \Sigma^*$. Then $w \in \inf(L)$. Considering
the word $w' = w^{|Q|+1}$, then $w' \in \inf(L)$ as well by the assumption. By the pigeonhole principal, an entire
copy of $w$ has to be read between some state $q$ and itself. Hence, $w \in \inf(L')$,
and $\inf(L') = \Sigma^*$.
 \end{proof}

The next result involves the definitions of pushdown automata
\cite{HU}, queue automata \cite{Harju2002278},
stack automata \cite{StackAutomata,Stack2}, and flip pushdown
automata \cite{flipPushdown}, which will not be defined
formally here. A pushdown automaton, queue automaton, or flip pushdown
automaton are reversal-bounded if there is a bound on the number
of switches between increasing and decreasing the size of the store.
For stack automata, it is defined to be reversal-bounded if there is a bound on the number of changes of direction of the read/write head
(which is influenced both by changing and reading the stack).
From the definitions, the following is straightforward.
\begin{proposition}
\label{languageSimulation}
Given $M$, a reversal-bounded pushdown automaton, then $L(M)$ can be accepted by a nondeterministic Turing machine with
a one-way read-only input, and a reversal-bounded read/write worktape.
This result also holds for $M$, a reversal-bounded queue machine,
a reversal-bounded stack machine, and a reversal-bounded $k$-flip pushdown automaton.
\end{proposition}

Then, from Proposition \ref{languageSimulation}, the following are true.
\begin{corollary}
It is decidable, given $L$ accepted by a reversal-bounded queue automaton, whether $L$ is dense.
\end{corollary}

\begin{corollary}
It is decidable, given $L$ accepted by a reversal-bounded stack automaton, whether $L$ is dense.
\end{corollary}

\begin{corollary}
It is decidable, given $L$ accepted by a reversal-bounded $k$-flip pushdown automaton, whether $L$ is dense.
\end{corollary}

Next, we will address undecidability. The first result applies
to most of the deterministic and nondeterministic families that have
an undecidable emptiness problem. The right marked concatenation of $L$
with a language $R$ is $L \$ R$ where $\$$ is not a letter of $L$ 
(but it can be a letter of $R$).
\begin{proposition}
Let ${\cal L}$ be a family of languages with an undecidable emptiness problem that is closed under right marked concatenation with regular languages. Then it is
undecidable whether $L \in {\cal L}$ is dense.
\label{nonemptiness}
\end{proposition}
\begin{proof}
Let $L \in {\cal L}$ over $\Sigma$. Let $\$$ be a new symbol, and let 
$\Sigma ' = \Sigma \cup \{\$\}$.
Let $L' = L \$ (\Sigma')^* \in {\cal L}$. 

If $L$ is empty, then $L'$ is empty too and so $L'$ is not dense.

If $L$ is not empty then say $w \in L$. Thus $w\$ (\Sigma')^*$ is a subset of $L'$ and so every word of $(\Sigma')^*$ is an infix of $L'$. Hence, $L'$ is dense.
\end{proof}

Next, we describe a construction which was essentially presented in \cite{Baker1974}.
\begin{construction}
\label{theconstruction}
Let $Z$ be a $\DTM$ operating on an initially blank tape.  
Assume that if $Z$ halts, it makes
$2k$ moves for some $k \ge 2$.

Let $w =  ID_1 \#  ID_3  \#  ID_5 \ldots \# ID_{2k-1}  \$ ID_{2k}^R  \#  \ldots \# ID_6^R  \# ID_4^R \#  ID_2^R$,
where the $ID_i$'s are configurations of $Z$.

Let $\Sigma$ be the alphabet over which $w$ is defined.

We construct a $1$-reversal $2$-stack real-time (i.e.\ no
$\lambda$-transitions) $\DPDA$ $M_1$ as follows, given input string $x$:
\begin{enumerate}
\item $M_1$ enters state $q_r$ if $x$ is not in the above format (the finite control can detect this).
\item If $x$ is in the above format, then $M_1$ enters state $q_r$  if  one of the following
is {\bf not} true:
\begin{enumerate}
\item $ID_1$ is not an initial configuration of $Z$ on blank tape,
\item $ID_{2k}$ is not a halting configuration of $Z$,
\item $ID_{i+1}$ is not a valid successor of $ID_i$ for some $i$.
\end{enumerate}
Otherwise, $M_1$ enters state $q_h$
\end{enumerate}
Thus, on any input $x$, $M_1$ only enters $q_h$  if $x$ is a halting computation of $Z$; otherwise $M$ enters $q_r$.
\end{construction}
As the halting problem is undecidable, it follows that the emptiness
problem is undecidable for one-way 1-reversal-bounded 2-stack real-time $\DPDA$s.

From this, and Proposition \ref{nonemptiness}, the following is true.
\begin{corollary}
It is undecidable, given a language $L$
accepted by a one-way deterministic $2$-stack real-time $\DPDA$
where both pushdowns are 1-reversal-bounded, whether $L$ is dense.
\end{corollary}

Undecidability is therefore immediate for deterministic Turing machines
with two 1-reversal-bounded worktapes, as 
each such worktape can simulate a pushdown.
\begin{corollary}
It is undecidable, given $L$ accepted by a deterministic
real-time Turing
machine with a one-way read-only input and two 1-reversal-bounded 
worktapes, 
whether $L$ is dense.
\end{corollary}
This shows that decidability changes between one and two reversal-bounded
worktapes for these kinds of Turing machines.

It is not clear whether (or not) the result above can be extended to 
unary languages accepted by one-way deterministic two pushdown machines. 
In fact, we conjecture that density is decidable for this model.
However, the next result shows that it is undecidable for unary
languages when nondeterminism is used.
\begin{proposition}
It is undecidable, given unary $L$ accepted by a one-way nondeterministic
real-time two pushdown machine, where both pushdowns are 1-reversal-bounded, 
whether $L$ is dense.
\end{proposition}
\begin{proof}
Starting with Construction \ref{theconstruction}, from 
$M_1$, now construct a $1$-reversal $2$-stack $\NPDA$ $M$ as follows.
$M$ simulates $M_1$ described  in the construction above by
reading symbol $a$ at each move of $M_1$ while nondeterministically
guessing the input to $M_1$. If $M_1$ lands in $q_h$ after
reading $a^k$ for some $k$, $M$ accepts any string $a^{k+i}$ for $i \ge 0$.
Clearly, $\inf(L(M)) = a^*$ if the DTM $Z$ halts on a blank tape. It $Z$
does not halt on a blank tape, $L(M)$ is empty and, hence,
$\inf(L(M))$ is also empty. It follows that $L(M)$ is dense if and
only if $Z$ halts on blank tape. Moreover, $M$ is real-time.
\end{proof}

There are some other interesting models with undecidable emptiness problems for which density will also be undecidable.
It is undecidable, given a $2\DCM(2)$ 
(a two-way $\DFA$ with two reversal-bounded 
counters) $M$ over a letter-bounded language, whether L(M) is empty \cite{Ibarra1978}. This model is also closed under right marked
concatenation. 
\begin{corollary}
It is undecidable, given a $2\DCM(2)$ $M$, whether $L(M)$ is dense. 
\end{corollary}

For $2\NCM$, i.e.\ two-way $\NFA$s with $k$ 
reversal-bounded counters for some $k$, it is known that all such
machines accepting unary languages can be effectively converted to $\NFA$s \cite{IbarraJiang}. Therefore:
\begin{proposition}
Given a unary $2\NCM$ $M$, it is decidable whether $L(M)$
is dense.
\end{proposition}
In contrast:  
\begin{proposition}  
It is undecidable, given a $2\DFA$ $M$ with one unrestricted counter
over a unary language,  whether $L(M)$ is dense.
\end{proposition}
\begin{proof} It is well known that it is undecidable, given a machine 
$Z$  with no input 
but with two (unrestricted) counters $C_1$ and $C_2$ that  are initially set to zero, 
whether $Z$  will halt  \cite{Minsky}. In fact, we may assume that if $Z$  does not halt, 
it goes into an infinite loop with increasing counter values.

Given such a machine $Z$,  we construct a $2\DFA$ $M$ which, when given an input
$w = a^k$   (with left and right end markers) simulates $Z$. $M$ has one counter $C_1$ 
to simulate $C_1$ of $Z$,  and uses its unary two-way input tape to  simulate  counter 
$C_2$ of $Z$.    

If, during the the simulation,  $M$  attempts  to go past the right end-marker,  $M$  
accepts $w$.   If , during the simulation, $Z$ halts, $M$ rejects.   Clearly, if $Z$ does not 
halt, $L(M) = a^*$. If $Z$ halts, $L(M)$ is finite.  It follows that $L(M)$ is dense if and only 
if $Z$ does not halt, which is undecidable.
\end{proof}

\begin{corollary} It is undecidable, given a $2\DFA$ $M$ with one unrestricted counter
over a unary language,  whether $L(M) = a^*$ (respectively, whether $L(M)$ is
finite).
\end{corollary}

It is open whether or not a language accepted by a $2\NCM$ with one reversal-bounded counter is empty, and density for this model is open as well

\section{Conclusions and Future Directions}

Here, we show that determining whether a language is dense is
decidable for nondeterministic Turing machines with a one-way
read-only input and one reversal-bounded worktape, and therefore
this is also decidable for one-way nondeterministic machines
with several different reversal-bounded data structures, such
as pushdowns, stacks, $k$-flip pushdowns, and queues.
However, if the number of reversal-bounded Turing tapes is 
increased to two, then it is undecidable, even if the machine
is deterministic, real-time, and $1$-reversal-bounded.

There are still several open problems related to determining
density. It is still unknown whether the density of $L$ is decidable, 
if $L$ is accepted by a one-way nondeterministic
(or deterministic) machine with multiple reversal-bounded counters. This
is decidable when there is only one counter.
The problem is also open for $L$ accepted
by one-way deterministic machines with one counter (no reversal-bound),
and for deterministic Turing machines with multiple reversal-bounded 
tapes accepting unary languages.

\biblio{density} 
\end{document}